\documentclass[11pt, reqno]{article}
\usepackage[letterpaper, total={6.5in, 9in}, 
]{geometry}

\usepackage{amsmath, amsthm, amscd, amsfonts, thmtools, amssymb, graphicx, xcolor, color}
\usepackage[bookmarksnumbered, colorlinks, plainpages=false]{hyperref}
\hypersetup{colorlinks=true,citecolor=magenta,linkcolor=blue} 

\allowdisplaybreaks
\usepackage[T1]{fontenc}
\usepackage{lmodern}
\usepackage{newtxmath,newtxtext}
\usepackage[page]{appendix}
\usepackage{array}
\usepackage{colortbl}
\newcolumntype{C}{>{$}c<{$}}
\newcolumntype{G}{>{\columncolor[gray]{0.8}$}c<{$}}
\usepackage{multirow}

\usepackage[
backend=biber,
style=alphabetic,
sorting=ynt
]{biblatex}
\usepackage{algorithm}
\usepackage{algpseudocode}
\floatname{algorithm}{Algorithm}

\newcommand{\diag}{\mathsf{Diag}}

\newcommand{\ba}{\boldsymbol{\alpha}}
\newcommand{\bb}{\boldsymbol{\beta}}
\newcommand{\bg}{\boldsymbol{\gamma}}

\newcommand{\cW}{\mathcal{W}}

\newcommand{\R}{\mathbb{R}}

\newcommand{\Z}{\mathbb{Z}}
\newcommand{\E}{\mathbb{E}}

\newcommand{\norm}[2]{\left\lVert #1 \right\rVert_{#2}}
\newcommand{\abs}[1]{\left| #1 \right|}
\newcommand{\para}[1]{\left( #1 \right)}
\newcommand{\set}[1]{\left\{#1\right\}}
\newcommand{\ip}[1]{\left\langle{ #1 }\right\rangle}

\newcommand{\Var}{\mathsf{Var}}

\newcommand{\eps}{\varepsilon}

\newtheorem{theorem}{Theorem}[section]
\newtheorem*{theorem*}{Theorem}
\newtheorem{lemma}[theorem]{Lemma}
\newtheorem{proposition}[theorem]{Proposition}

\theoremstyle{definition}
\newtheorem{definition}[theorem]{Definition}

\theoremstyle{remark}
\newtheorem{remark}[theorem]{Remark}
\numberwithin{equation}{section}

\newif\ifnotanonymous
\notanonymoustrue

\begin{document}
\setcounter{page}{1}

\title{Improved 
Sparse Recovery
for Approximate Matrix Multiplication}

\ifnotanonymous
    \author{
    Yahel Uffenheimer \thanks{Hebrew University of Jerusalem. Supported by ERC Starting Grant (CODY 101039914).} 
    \and 
    Omri Weinstein \thanks{Hebrew University of Jerusalem. Supported by ISF grant \#3011005535 and ERC Starting Grant (CODY 101039914).}}
\else
    \author{Anonymous Authors}
\fi

\maketitle

\begin{abstract}
We present a simple randomized algorithm for approximate matrix multiplication (AMM) whose error scales with the \emph{output} norm $\|AB\|_F$. Given any $n\times n$ matrices $A,B$ and a runtime parameter $r\leq n$, the algorithm produces in $O(n^2(r+\log n))$ time,  a matrix 
$C$ with total squared error $\E[\|C-AB\|_F^2]\le (1-\frac{r}{n})\|AB\|_F^2$, per-entry variance $\|AB\|_F^2/n^2$ and bias $\E[C]=\frac{r}{n}AB$. Alternatively, the algorithm can compute an \textit{unbiased} estimation with expected total squared error $\frac{n}{r}\norm{AB}{F}^2$, recovering the state-of-art AMM error obtained by Pagh's TensorSketch algorithm \cite{pagh13}. Our algorithm is a log-factor faster.

The key insight in the algorithm is a  new variation of pseudo-random rotation of the input matrices (a Fast Hadamard Transform with asymmetric diagonal scaling), which redistributes the Frobenius norm of the \emph{output} $AB$ uniformly across its entries. 
\end{abstract} 

\section{Introduction}

Matrix multiplication is a fundamental operation across all fields of science and technology. Most notably, matrix multiplications are the backbone and computational bottleneck of training and inference of deep neural networks, where both forward and backpropagation rely on giant matrix multiplications (e.g., multiplying $16K\times 16K$ matrices is nowadays considered a prerequisite in any LLM \cite{Li24LLMs,overview_of_llms}). For trillion-scale parameter models, the difference between na\"ive matrix-multiplication ($\sim n^3$ time) versus the information-theoretic lower bound $\Omega(n^2)$, is a fundamental concern.

The asymptotics and hence impracticality of \emph{fast matrix multiplication}\footnote{Asserting that the product $AB$ of two $n\times n$ real matrices can be 
computed in $O(n^\omega) \sim O(n^{2.37})$ time \cite{S69,w12,l14, duan2023faster}.} (FMM) algorithms like Strassen\footnote{Here we refer to Strassen-like algorithms, using a recursive application of some basic bilinear algorithm. The constants in most of these algorithms are too large to be practical. Strassen's original algorithm is a unique outlier.} \cite{S69}, initiated a long line of research on \emph{approximate matrix multiplication} (AMM), which studies the best \emph{speed-accuracy} tradeoff achievable by ``combinatorial'' algorithms, that avoid divide-and-conquer and have \emph{non-asymptotic} sub-cubic runtime. 
More formally, for a prescribed parameter $r < n$, the  goal is to produce, in $\tilde{O}(rn^2)$ time, a matrix $C \in \R^{n\times n}$, which $\eps$-approximates $AB$ in the Frobenius norm, where $\eps$ is a decreasing function of $r$ (meaning that $\eps$ tends to $0$ as $r\to n$).

Essentially all known AMM algorithms use randomized \emph{sketching or sampling}  techniques \cite{sar06, drineas06, pagh13, magen10, cw13, CNW16, CL99}, and the state-of-art after more than 20 years of research is a \emph{linear} speed-accuracy tradeoff: 
\begin{align}\label{eq_AMM_sota}
    \|C - AB\|^2_F 
    \leq  
    O\left( \min\left\{ \frac{1}{r}\norm{A}{F}^2\cdot \norm{B}{F}^2 \; , \; \frac{n}{r}\|AB\|^2_F \right\} \right),  
    \tag{AMM error}
\end{align}
where $C$ must be produced in $\tilde{O}(rn^2)$ (randomized) time.  The first error term can be obtained via the standard ``sketch-and-solve'' algorithm (e.g., CountSketch \cite{countsketch, cw13}), while the second one is obtained by a clever \emph{output-sensitive} variation of CountSketch, using FFT (TensorSketch \cite{pagh13}). The two bounds in \eqref{eq_AMM_sota} 
are generally incomparable, but they coincide for the (hardest known) distribution of random Gaussian (or Rademacher) matrices.  

A conceptual limitation of all aforementioned AMM algorithms (except \cite{CL99} which applies only to nonnegative matrices), is that they use \emph{compression} techniques, i.e., compress each matrix (using low rank projections or subset sampling) and compute the product on the compressed representations. In data-driven applications, most notably LLM training and inference, this is a severe limitation since compression crucially decreases the number of trainable parameters (see \cite{STL25} and references therein).
Moreover, a recent result of \cite{OP25} proves that for compression-based algorithms, in a setup where Alice and Bob can send $O(rn)$ bits representing their respective input matrices $A,B \in \R^{n\times n}$ to a ``referee'' who must then compute the output based on their messages, the error bound in \eqref{eq_AMM_sota} is tight for  random Gaussian (and Rademacher) matrices.

The idea of using a fast orthogonal rotation matrix like the Walsh-Hadamard transformed, with randomized signs, was first proposed by \cite{fast-jl} when introducing the fast-JL transform. This technique has evolved into the sub-sampled randomized Hadamard transform (SRHT) and was adapted for randomized linear algebra. Specifically it has been used for low-rank approximation and matrix sketching \cite{tropp11, bout13}. We highlight \cite{bout13}, as they were the first to apply the technique to the AMM task. The main difference between SRHT and our approach is that we do not apply sub-sampling until after we rotate both matrices. Instead of sketching both matrices and computing the product of sketches, we transform them to a more convenient form and then sketch their product directly. Our transformation acts on the matrix space and not on the vector space. This difference allows a fast runtime without compressing the matrices (and also escapes the low-rank constraint of common sketching algorithms).


\section{Fast Walsh-Hadamard Sketch}
\subsection{Preliminaries}
Denote $[n]=\set{0,\ldots,n-1}$. The Frobenius norm of a square $n\times n$ matrix $A$ is defined as $\norm{A}{F}^2=\sum_{i,j\in [n]}\left|A_{i,j}\right|^2$. The $2^k$-th Walsh-Hadamard transform (WHT), denoted by $H_{2^k}$, is defined recursively: $$H_2=\frac{1}{\sqrt{2}}\cdot \begin{pmatrix}
    1 & 1\\ 1& -1
\end{pmatrix}\quad,\quad H_{2^{k+1}}=H_{2^k}\otimes H_{2^k}.$$More explicitly, for $i,j\in [2^k]$ it holds $$(H_{2^k})_{i,j}=\frac{1}{\sqrt{2^k}}(-1)^{\ip{i,j}_b}\quad \text{where}\quad \ip{i,j}_b=\bigoplus_{t=0}^{k-1}b(i)_t\oplus  b(j)_t,$$letting $b(i)$ denote the binary representation of $i$ and $\oplus$ denote binary addition (i.e., addition in $\mathbb{F}_2$). When $k$ is clear from context, we drop the subscript and write $H$. 
We note that $H$ has the following important properties -- it is a unitary symmetric involution, that is: $$H=H^{\top}\quad,\quad H^{-1}=H\quad,\quad \norm{Hx}{2}=\norm{x}{2}.$$We also note that $H_{2^k}$ is the Discrete Fourier transform of the group $\Z_2^k$. For a vector $a$ of size $n$, let $D_a=\diag(a)$. Note that for every couple of vectors $a,b$ of size $n$ and a matrix $X$ of size $n\times n$ it holds $$(ab^{\top})\odot X=D_aXD_b,$$using the vec-trick, where $\odot$ denotes the element-wise product of matrices.

\subsection{Sketch Idea}
Note that for random Rademacher matrices $A,B$, since they are highly ``balanced'' in the sense that mass (i.e., the magnitude of the entries) is uniformly distributed, one can approximate $AB$ by computing any fixed set of $rn$ output entries, and fill the rest with zeros. This is a biased estimator achieving $\frac{n-r}{n}\cdot \norm{AB}{F}^2$ squared Frobenius error in expectation (over the randomness of $A,B$).

For arbitrary fixed $A,B$, the matrices may be unbalanced (have most of the mass concentrated on a small number of entries), and thus following the same strategy may lead to a very large and uncontrolled error, as the choice of wrong entries may have dire effects. Our idea, inspired by the Fast-JL Transform \cite{fast-jl}, is to precondition the matrices in order to obtain new ``\textit{pseudo-random}'' matrices, for which the above na\"ive algorithm should work well. Intuitively, the uncertainty principle (for the Fourier transform, see \cite{wig20}) mitigates severe imbalances in the original matrices. For this to work, we design the preconditioning to be invertible.

We note that by the vec-trick, for a matrix $X$ of size $2^k\times 2^k$ it holds $$\mathrm{vec}( H_{2^k}XH_{2^k}) =H_{2^{k+1}} \mathrm{vec}({X}).$$Therefore we can view the preconditioning done in the algorithm (presented next), as a transformation in the matrix space, compared with standard sketching algorithms, where the sketch is applied on each column / row of the matrices separately (at least conceptually).

\subsection{The Algorithm}
Let $A,B$ be given $n\times n$ matrices where we assume $n=2^k$ for some $k$. Let $H=H_{2^k}$ denote the WHT. Let $r\le n$ be a given parameter.

\begin{definition}
    For $\ba,\bb\in \set{\pm 1}^n$ define $\cW_{\ba,\bb}$ to be an operator on $n\times n$ matrices defined by $\cW_{\ba,\bb}(A)=HD_{\ba}AD_{\bb}H$. It is an invertible unitary linear operator with inverse given by $\cW_{\ba,\bb}^{-1}(B)=D_{\ba}HBHD_{\bb}$.
\end{definition}

\begin{lemma}\label{lem:1}
    For any $\ba,\bb,\bg\in \set{\pm1}^n$ and matrices $A,B$, it holds $\cW_{\ba,\bb}(AB)=\cW_{\ba,\bg}(A)\cdot \cW_{\bg,\bb}(B)$ and $\cW_{\ba,\bb}^{-1}(AB)=\cW_{\ba,\bg}^{-1}(A)\cdot \cW_{\bg,\bb}^{-1}(B)$.
\end{lemma}
\begin{proof}
    Indeed, $\cW_{\ba,\bb}(AB)=HD_{\ba}ABD_{\bb}H=HD_{\ba}AD_{\bg}HHD_{\bg}BD_{\bb}H=\cW_{\ba,\bg}(A)\cdot \cW_{\bg,\bb}(B)$ using the fact $D_{\bg}^2=I$ and $H^2=I$. Similarly for the inverse.
\end{proof}

\begin{algorithm}
    \caption{Approximate Matrix Multiplication}
    \begin{algorithmic}[1]
        \Require{$n\times n$ matrices $A,B$.}
        \Ensure{an approximation $C$ for $AB$.}

        \State{Draw random sign vectors $\ba,\bb,\bg\in \set{\pm1}^n$.}\label{item:1}
        \State{Compute $A'=\cW_{\ba,\bg}(A),B'=\cW_{\bg,\bb}(B)$.}\label{item:2}
        \State{Choose any set of $n\cdot r$ indices $(i,j)\in [n]^2$ and compute $A_{i,:}'\cdot B_{:,j}'$. Store the results in a matrix $C'$ initialized to zero.}\label{item:3}
        \State{\Return $\cW_{\ba,\bb}^{-1}(C')$.}\label{item:4}
    \end{algorithmic}
\end{algorithm}

\textbf{Runtime.} Drawing random sign vectors requires $O(n)$ time. To compute $A',B'$, note that $\mathrm{vec}(HXH)=(H\otimes H)\mathrm{vec}(X)$ can be computed using the fast WHT algorithm in time $O(n^2\log n)$, while a product with a diagonal matrix can be computed in time $O(n^2)$. Hence step \ref{item:2} takes $O(n^2\log n)$ time. For step \ref{item:3}, a trivial calculation takes $O(rn\cdot n)$ time. At last, step \ref{item:4} is the same as step \ref{item:2}. We conclude that $C$ can be computed in time $O(n^2(r+\log n))$.

\begin{remark}
    A few remarks are in place:
\begin{itemize}
    \item Compared to previously known algorithms, like that of \cite{pagh13}, the runtime is a logarithmic factor faster.
    \item The specifics of step \ref{item:3} are left unspecified on purpose. In the analysis to follow we will assume a \textit{random} subset of size $rn$ is chosen, but different methods are possible, maybe with some improvements to the analysis. Since the error analysis (shown next) is only concerned with the marginal distributions, it is blind to other effects of the implementation of step \ref{item:3}.
    \item The vector $\bg$ is canceled out and so has no algorithmic effect (and can be removed). However, when analyzing the properties of $A',B'$ on their own, the presence of $\bg$ does have significance.
    \item Setting $r=n$ we have $C'=A'B'$, and so by definition of $C$ and \autoref{lem:1}, $C=AB$. This is unusual in sketching algorithms, which usually cannot recover the result exactly due to likely hashing collisions.
    \item If the $n\cdot r$ positions are chosen to be specific rows or columns of $A'B'$ (say the first $r$ rows), then the rank of $C'$ is at most $r$. Since $\cW_{\ba,\bb}$ is linear and invertible, the resulting rank of $C$ is $r$ too. However, if the positions are chosen in an unstructured way (say, a random set of size $rn$), we don't have such a guarantee. This doesn't change the analysis, but it does mean the error isn't inherently lower bounded by the sum of the lower $n-r$ singular values of $AB$.
\end{itemize}
\end{remark}

\subsection{Error Analysis}
\begin{proposition}
    Assuming the calculated positions are randomly chosen (uniformly), it holds $\E[C]=\frac{r}{n}AB$.
\end{proposition}
\begin{proof}
    By assumption, $\E[C^{\prime}_{i,j}]=\frac{r}{n}\cdot (A'B')_{i,j}$ where the probability is taken over the choice of indices. Therefore, taking the expectation over both indices and signs, 
    \begin{align*}
        \E[C_{i,j}]&=\E[(D_{\ba}HC'HD_{\bb})_{i,j}]
         =\E\left[\ba_{i}\bb_{j}\cdot \sum_{k,\ell}H_{i,k}C^{\prime}_{k,\ell}H_{\ell,j}\right] 
        = \sum_{k,\ell}H_{i,k}H_{\ell,j}\cdot \E[\ba_i\bb_j\cdot C^{\prime}_{k,\ell}].
    \end{align*}
    Using conditioned expectation and \autoref{lem:1}, $$\E[\ba_i\bb_jC^{\prime}_{k,\ell}]=\frac{r}{n}\cdot \E[\ba_{i}\bb_j(A'B')_{k,\ell}]=\frac{r}{n}\sum_{p,q}H_{k,p}H_{q,\ell}\cdot (AB)_{p,q}\cdot \E[\ba_i\bb_j\cdot \ba_{p}\bb_q].$$
    Note that $\E[\ba_i\bb_j\cdot \ba_{p}\bb_q]=0$ unless $i=p,j=q$ in which case it is equal to $1$. Therefore $\E[\ba_i\bb_jC^{\prime}_{k,\ell}]=\frac{r}{n}H_{k,i}H_{j,\ell}(AB)_{i,j}$, implying that $$\E[C_{i,j}]=\sum_{k,\ell} H_{i,k}H_{\ell,j}H_{k,i}H_{j,\ell}\cdot \frac{r}{n}\cdot (AB)_{i,j}.$$We finish by recalling that $H$ is symmetric with elements of magnitude $1/n$, thus $$\E[C_{i,j}]=\frac{n^2}{n^2}\cdot \frac{r}{n}\cdot (AB)_{i,j}=\frac{r}{n}(AB)_{i,j}.$$
\end{proof}

\begin{proposition}
    It holds $\Var(\norm{C-AB}{F})\le \E[\norm{C-AB}{F}^2]=\frac{n-r}{n}\cdot \norm{AB}{F}^2.$
\end{proposition}

\begin{proof}
    Since $\cW_{\ba,\bb}$ is unitary, we have
    \begin{align}\label{eq:1}
    \norm{C-AB}{F}^2 =\norm{\cW_{\ba,\bb}^{-1}(C'-A'B')}{F}^2=\norm{C'-A'B'}{F}^2.
    \end{align}
    For any $i,j\in [n]$, by \autoref{lem:1} we have $$(A'B')_{i,j}=\sum_{k,\ell} H_{i,k}H_{\ell,j}\cdot \ba_{k}\bb_{\ell}\cdot (AB)_{k,\ell}.$$Therefore, denoting $R_{k,\ell}=H_{i,k}H_{\ell,j}\cdot (AB)_{k,\ell}$ we have: 
    $$\E[(A'B')_{i,j}^2]=\sum_{k,l,p,q}R_{k,l}R_{p,q}\cdot \E\left[\ba_{k}\ba_{p}\cdot \bb_{\ell}\bb_{q}\right].$$The expectation over the signs is $0$ unless $k=p$ and $\ell=q$, in which case it is $1$. Thus the sum collapses,
    $$\E[(A'B')_{i,j}^2]=\sum_{k,\ell} R_{k,\ell}^2 =\frac{1}{n^2}\sum_{k,\ell}(AB)_{k,\ell}^2=\frac{\norm{AB}{F}^2}{n^2}.$$
    Conditioned on any fixed set of indices $F\subset [n]^2$ of size $rn$, 
    \begin{align}\label{eq:2}
        \E[\norm{C'-A'B'}{F}^2\mid F\text{ was chosen}]=\sum_{(i,j)\notin F}\E[(A'B')_{i,j}^2]= (n^2-rn)\cdot \frac{\norm{AB}{F}^2}{n^2}.
    \end{align}
    Returning to \autoref{eq:1}, we conclude $$\E[\norm{C-AB}{F}^2]=\E[\norm{C'-A'B'}{F}^2]=\frac{n-r}{n}\cdot \norm{AB}{F}^2.$$
\end{proof}

\begin{remark}
    We can output an unbiased estimator by scaling step \ref{item:4} to be $C=\frac{n}{r}\cdot \cW_{\ba,\bb}^{-1}(C')$. In this case, if the index $(i,j)$ was chosen to be computed in step \ref{item:3}, then $C_{i,j}^{\prime}=\frac{n}{r}\cdot (A'B')_{i,j}$, which implies $$\abs{C_{i,j}^{\prime}- (A'B')_{i,j}}^2=\abs{\para{\frac{n}{r}-1}\cdot (A'B')_{i,j}}^2.$$
    Noting that $(n/r-1)^2=\frac{n^2}{r^2}-2\frac{n}{r}+1$, \autoref{eq:2} becomes
    \begin{align*}\label{eq:3}
    \E[\norm{C'-A'B'}{F}^2] &=\sum_{i=0}^{n-r-1}\sum_{j=0}^{n-1}\E[(A'B')_{i,j}^2]+\sum_{i=n-r}^{n-1}\sum_{j=0}^{n-1}\para{\frac{n}{r}-1}^2\cdot \E\left[(A'B')_{i,j}^2\right]
    \\ &=\frac{1}{n^2}\para{n(n-r)+nr\cdot \para{\frac{n^2}{r^2}-2\frac{n}{r}+1}}\cdot \norm{AB}{F}^2
    \\ &=\para{1-\frac{r}{n}+\frac{n}{r}-2+\frac{r}{n}}\cdot \norm{AB}{F}^2\le \frac{n}{r}\cdot \norm{AB}{F}^2.
    \end{align*}
    This is exactly the variance bound achieved by \cite{pagh13}, which is the state-of-the-art for unbiased estimators. Thus we recover the same result, with different tools.
\end{remark}

\begin{remark}
    Note that the trivial algorithm of randomly choosing $rn$ coordinates and computing their values, while setting all other values to null, achieves the same global guarantee, meaning $\E[C]=\frac{r}{n}AB$ and $\E[\norm{C-AB}{F}^2]=\frac{n-r}{n}\cdot \norm{AB}{F}^2$. However, the main difference from our approach, or Pagh's approach \cite{pagh13}, is the per-entry variance. In this na\"ive algorithm, we have $$\E[(C_{ij}-(AB)_{ij})^2]=\frac{n-r}{n}\cdot (AB)_{ij}^2,$$so the variance is related to the magnitude of the $i,j$-th output value. In our algorithm, the per-entry variance is $\frac{n-r}{n}\cdot \frac{\norm{AB}{F}^2}{n^2}$, i.e., the average squared value of the output matrix. If we multiply by $\frac{n}{r}$ to obtain unbiased estimators, the per-entry variance becomes $\le \frac{n}{r}\cdot (AB)_{ij}^2$ for the na\"ive algorithm and $\le \frac{n}{r}\cdot \frac{\norm{AB}{F}^2}{n^2}=\frac{\norm{AB}{F}^2}{nr}$ for our algorithm. This is the same as the per-entry variance of Pagh's algorithm. This is useful for sparse applications, as demonstrated in Pagh's original paper.
\end{remark}

\section{Discussion}
This simplification for a known result was recovered when trying to formalize and prove that random Gaussian matrices (or Rademacher matrices), are the hardest case for approximate matrix multiplication algorithms. We hoped to reduce the worst-case to the average case, thus allowing one to obtain algorithms for any (real) matrices from an algorithm that works well (on average) for random matrices (Gaussians or Rademacher).

In the finite field case, a recent work of \cite{zmir25} shows such a reduction, using cryptographic assumptions. However, these do not extend to the real case. The main failure is the fact real distributions cannot be shift invariant. Unfortunately, the random rotations we present here are not sufficient for such a reduction. First, it is unclear how to define in-distinguishability from a matrix distribution in this context. Second, the rotation doesn't change the spectral properties of the matrices (since it is unitary), which is in our view a prerequisite for any meaningful reduction (and a distinguishing feature).

Another feature of our algorithm is that it can be amplified to an exact algorithm. In other words, given black box access to the algorithm, running $O(\frac{n}{r}\ln(1/\eps))$ independent iterations suffice to obtain a matrix $C$ which satisfies $\norm{C-AB}{F}^2\le \eps$ assuming $\norm{AB}{F}=1$. While it has no practical use, since one can just compute the exact product in time $O(n^3)$, it is a unique feature of this algorithm. As mentioned above, most sketching algorithms inherently corrupt the data, while our algorithm intuitively ``peels off'' a uniform layer. It is not even clear if any algorithm achieving the same guarantees can be amplified to an exact algorithm. It is interesting to understand what type of constraints do these guarantees impose on a (randomized) bilinear algorithm.

We finish the discussion by pointing out that more improvements might be made to this algorithm. In particular, switching out step \ref{item:3} with a smarter algorithm might give better results. Any algorithm that works well for Gaussian or Rademacher matrices and relies only on first and second moment properties might extend this result.

\ifnotanonymous
\fi

\printbibliography
\end{document}